\documentclass[final,hidelinks]{siamart251216}

\usepackage{amsfonts,amssymb,latexsym,amsxtra,mathrsfs}
\usepackage{xcolor}
\usepackage{float}
\usepackage{booktabs,array}
\usepackage{enumitem}

\newsiamthm{conjecture}{Conjecture}
\newsiamthm{open problem}{Open Problem}
\newsiamremark{example}{Example}
\newsiamremark{remark}{Remark}
\counterwithout{equation}{section}
\counterwithout{table}{section}
\counterwithout{figure}{section}

\title{Beyond Polynomials: Optimal Locally Recoverable Codes from Good Rational Functions}
\author{%
  Hengfeng Liu\thanks{School of Mathematics, Southwest Jiaotong University, Chengdu,  China (\email{hengfengliu@163.com}).}
  \and
  Sihem Mesnager\thanks{Department of Mathematics, University of Paris VIII, 93526 Saint-Denis, France; CNRS, UMR 7539 LAGA, University of Paris XIII,  France (\email{smesnager@univ-paris8.fr}).}
  \and
  Chunming Tang\thanks{School of Information Science and Technology, Southwest Jiaotong University, Chengdu,  China (\email{tangchunmingmath@163.com}).}
  \and
  Xuemin Zheng\thanks{School of Mathematical Sciences, China West Normal University, Nanchong, China (\email{2455891543@qq.com}).}
}
\headers{LRCs Beyond Polynomials}{H. Liu, S. Mesnager, C. Tang, and X. Zheng}
\date{}

\hypersetup{%
  pdftitle={Beyond Polynomials: Optimal Locally Recoverable Codes from Good Rational Functions},
  pdfauthor={Hengfeng Liu, Sihem Mesnager, Chunming Tang, Xuemin Zheng}
}

\begin{document}
\maketitle
\raggedbottom

\begin{abstract}
Locally recoverable codes (LRCs) play an important role in distributed and cloud storage systems. The celebrated Tamo--Barg construction of optimal LRCs is based on the notion of \emph{good polynomials}. In this paper, we introduce the notion of \emph{good rational functions} as a natural generalization of good polynomials and develop a more general framework for constructing optimal LRCs. Since the attainable code lengths depend on the number of fibers of maximal size, we study this quantity using algebraic function field theory and Galois theory, identifying it with the number of totally split rational places in the associated function field extension. By analyzing the Galois closure, we establish an exact counting formula and derive sharp upper bounds for the number of ramified places and the genus, leading to explicit lower bounds for the number of totally split rational places. When the extension is Galois, we further characterize the number of totally split rational places in terms of the action of the Galois group on the projective line. Explicit constructions of good rational functions in both Galois and non-Galois settings are also presented. These rational functions yield optimal LRCs with lengths exceeding those attainable from good polynomials of the same degree, demonstrating the advantage of our framework over the classical Tamo--Barg construction.		
		\end{abstract}



	\section{Introduction}
	Let $q$ be a prime power, and let $\mathbb{F}_q$ denote the finite field with $q$ elements. Let $n,k,r$ be positive integers. An $(n,k,r)$-locally recoverable code (LRC) $\mathcal{C}$ is defined as a $k$-dimensional linear subspace of $\mathbb{F}_{q}^{n}$ with the property that any erased coordinate of a codeword $c \in \mathcal{C}$ can be recovered by accessing at most $r$ other coordinates. 
The parameter $r$ is referred to as the \emph{locality} of $\mathcal{C}$, and it measures the efficiency of the local repair process. In particular, smaller values of $r$ correspond to more efficient recovery procedures, which is a key requirement in applications such as distributed storage systems. Locally recoverable codes (LRCs) have attracted significant attention due to their fundamental role in distributed and cloud storage systems, where efficient data recovery is crucial \cite{BargTamoVladut2015}, \cite{Gopalan H S Y2012}, \cite{Micheli 2020}, \cite{Tamo Barg 2014}. Over the past decade, LRCs have been extensively studied from both theoretical and practical perspectives, leading to a rich literature on their constructions, bounds, and structural properties \cite{Cadambe Mazumdar2015}, \cite{Cai Miao Schwartz Tang}, \cite{Duke F M2022}, \cite{Jin Ma Xing 2020}, \cite{Li Ma Xing2019}, \cite{Luo Cao 2021}, \cite{Micheli 2020}, \cite{Tamo Papailiopoulos2016}, \cite{Xing Yuan 2022}, as well as recent advances \cite{Fang Tao Fu Chen Xia}, \cite{Heng2024}, \cite{Heng2025}, \cite{Luo Ezerman Ling2023}, \cite{Luo Chen Ezerman Ling2025}, \cite{Ma Xing 2023}, \cite{Sharma2025}, \cite{Tan2023}.

A fundamental limitation on the parameters of LRCs was established by Gopalan \textit{et al.} \cite{Gopalan H S Y2012}, who proved that the minimum distance $d\left (  \mathcal{C} \right )$ satisfies
\[
d\left (  \mathcal{C} \right ) \leq n - k - \left\lceil \frac{k}{r} \right\rceil + 2.
\]
This inequality is known as the \emph{Singleton-type bound} for LRCs. Codes attaining this bound are called \emph{optimal}, and constructing such codes with flexible parameters remains a central challenge in coding theory.

Among the various construction paradigms, a particularly influential framework was introduced by Tamo and Barg \cite{Tamo Barg 2014}, where the notion of a \emph{good polynomial} plays a pivotal role. Intuitively, a polynomial is considered good if it exhibits a strong regularity property on suitably chosen subsets of $\mathbb{F}_q$. More precisely, a polynomial $f\in \mathbb{F}_q[X]$ is said to be $(r,l)$-good ($l>0$) if 
\begin{itemize}
	\item $f$ has degree $r+1$,
	\item there exist $l$ disjoint subsets $A_{1},\ldots, A_{l}$ in $\mathbb{F}_q$, with $|A_{i}|=r+1$ ($1\le i\le l$), such that $f$ is constant on each $A_{i}$.
\end{itemize} 

This structural property allows one to partition evaluation points into local groups, which is the key mechanism underlying locality.

The Tamo--Barg construction is realized via generalized Reed--Solomon codes. Assume $r\mid k$. For a message $a\in \mathbb{F}_{q}^{k}$, written as $a_{ij}$ with $0\le i\le r-1$ and $0\le j\le k/r-1$, the encoding polynomial is given by
\[
f_{a}(x) = \sum_{i=0}^{r-1} \sum_{j=0}^{\frac{k}{r}-1} a_{ij} f(x)^{j}x^{i}.
\]
The LRC is obtained by evaluating $f_{a}$ on the set \( A = \bigcup_{i=1}^{l} A_{i} \), namely,
\[
\mathcal{C} = \left\{ \big( f_{a}(x) : x \in A \big) \mid a \in \mathbb{F}_{q}^{k} \right\}.
\]

Such an \( (n, k, r) \)-LRC is optimal, with parameters $n=(r+1)l$ and $k=tr$ ($t\le l$). In addition, \cite{Tamo Barg 2014} provided explicit constructions of good polynomials arising from additive and multiplicative subgroups of $\mathbb{F}_{q}$, under certain arithmetic constraints on $(q,r)$. These subgroup-based constructions impose arithmetic conditions on $r+1$, such as $r+1$ dividing $q-1$ or $q$.

These constraints motivated further generalizations. In this direction, Liu, Mesnager, and Chen \cite{Liu Mesnager Chen 2018} proposed a broader construction method based on function composition. Moreover, Liu, Mesnager, and Tang \cite{Liu Mesnager Tang 2020} showed that good polynomials can be constructed using Dickson polynomials and their compositions, thereby enriching the available families.

To construct LRCs with a wide range of lengths and dimensions, it is crucial to identify good polynomials that are constant on as many $(r+1)$-subsets as possible. To quantify this property, for a polynomial $f\in\mathbb{F}_{q}[X]$, one defines
\begin{equation}\label{size}
l_{f}=\#\left \{ t_{0}\in\mathbb{F}_{q}\mid f-t_{0} \; \text{has $\deg(f)$ distinct roots in $\mathbb{F}_{q}$} \right \}.
\end{equation}
This quantity measures the number of fibers of maximal size under the map induced by $f$. It is immediate that $l_{f}\le q/\deg(f)$.

A key breakthrough in understanding this parameter was achieved by Micheli \cite{Micheli 2020}, who interpreted $l_{f}$ within the framework of algebraic function field theory. More precisely, for a separable polynomial $f$, the quantity $l_{f}$ coincides with the number $\#T^1_{\mathrm{split}}(f)$ of totally split rational places in the global function field extension $\mathbb{F}_{q}(x)/ \mathbb{F}_{q}(t)$ defined by $f(x)=t$. Under the assumption that the Galois closure has full constant field $\mathbb{F}_q$, the density argument in \cite{Micheli 2020} gives
\[
\#T^1_{\mathrm{split}}(f) = \frac{q}{\#G_{f}} + O(\sqrt{q}),
\]
where $G_f$ is the arithmetic monodromy group of $f$. The implied constant depends on the associated extension. If the constant field of the Galois closure is larger than $\mathbb{F}_q$, there are no totally split rational places.

This viewpoint reveals that the construction of good polynomials is closely related to controlling the size and structure of their associated Galois groups.

In Micheli's work \cite{Micheli 2020}, several classes of good polynomials of low degree were constructed, together with lower bounds on $l_{f}$. This line of research was further developed by Chen, Mesnager, and Zhao \cite{Chen Mesnager zhao2021}, who completely characterized good polynomials of low degree over finite fields and determined their corresponding values of $l_{f}$. Subsequently, Dukes, Ferraguti, and Micheli \cite{Duke F M2022} studied the optimal monodromy groups in degrees up to $5$ and obtained explicit estimates for the maximal number of full rational fibers.

Building on these developments, Chen and Mesnager \cite{Chen Mesnager 2022} refined the Galois-theoretic framework by characterizing polynomials with minimal Galois groups and deriving explicit formulas for $l_{f}$ in several important cases, most notably for Dickson polynomials.

Rational maps have also been used to construct LRCs through fiberwise interpolation \cite{Munuera Tenorio 2018}, while automorphism groups of rational function fields have yielded optimal LRCs of length up to $q+1$ \cite{Jin Ma Xing 2020}. Our interest here is in counting full rational fibers of a single rational function on the projective line and comparing the resulting construction with good polynomials of the same degree.

\subsection*{Motivations and objectives of this paper}

We formalize the notion of \emph{good rational functions} as a natural generalization of good polynomials and give an explicit interpolation construction based on their full rational fibers. Unlike the classical polynomial-based construction, our approach allows evaluation points to be chosen from the projective line $\mathbb{P}^{1}(\mathbb{F}_{q})$, including the point at infinity. We show that good rational functions can be used to construct optimal LRCs, with the attainable code lengths determined by the number of fibers of maximal size.

For a rational function $h$, this number coincides with $\#T^1_{\mathrm{split}}(h)$, the number of totally split rational places in the extension $\mathbb{F}_{q}(x)/\mathbb{F}_{q}(h(x))$. We study this quantity using algebraic function field theory and Galois theory. By passing to the Galois closure, we obtain an exact counting formula and establish sharp upper bounds for the number of ramified places and the genus of the Galois closure, leading to explicit lower bounds for $\#T^1_{\mathrm{split}}(h)$. When the extension is Galois, we further characterize this number in terms of the action of the Galois group on the projective line. Finally, we give explicit constructions of good rational functions in both Galois and non-Galois settings. For suitable finite fields, these constructions yield optimal LRCs with lengths exceeding those attainable from good polynomials of the same degree.

\subsection*{Organization of the paper}

The paper is organized as follows. Section~\ref{pre} recalls the necessary background on algebraic function fields. Section~\ref{LRCs from good rational functions} introduces good rational functions and presents a construction of optimal LRCs based on their fibers on the projective line. Section~\ref{how good} relates the number of fibers of maximal size to the number of totally split rational places in the associated function field extension. We derive an exact counting formula, sharp upper bounds for the number of ramified places and the genus of the Galois closure, and lower bounds for the number of totally split rational places. Section~\ref{Gal case} characterizes this number in the Galois case through group actions and presents explicit Galois and non-Galois constructions, together with comparisons with good polynomials of the same degree. Finally, Section~\ref{summary} concludes the paper and discusses further directions.

\section{Preliminaries}\label{pre}

In this section, we recall some fundamental notions and results from the theory of algebraic function fields and their extensions, which will serve as essential tools throughout the paper. We also fix notation and briefly discuss the key concepts that will be used in the subsequent sections.

Let \( K \) be a perfect field. A rational function field over \( K \) is \( K(x) \), where \( x \) is transcendental over \( K \). A function field $F/K$ is defined as a finite algebraic extension of $K(x)$.

The set 
\[
\tilde{K}:=\{z\in F \mid z \text{ is algebraic over } K\}
\]
forms a subfield of $F$, called the \emph{constant field} of $F/K$. Clearly, one has $K \subseteq \tilde{K} \subsetneq F$, and the extension $F/\tilde{K}$ is itself a function field over $\tilde{K}$. We say that $K$ is the \emph{full constant field} of $F$ if $K = \tilde{K}$. This condition will be particularly relevant in our later considerations involving global function fields.

The most basic objects in the theory of function fields are valuation rings and places, which play a role analogous to primes in number fields.

A \emph{valuation ring} of \( F/K \) is a ring \( \mathcal{O} \subseteq F \) such that \( K \subsetneq \mathcal{O} \subsetneq F \) and for every \( z \in F\setminus\{0\} \), either \( z \in \mathcal{O} \) or \( z^{-1} \in \mathcal{O} \). A \emph{place} \( P \) of \( F/K \) is the maximal ideal of some valuation ring \( \mathcal{O}_P \) of \( F/K \). We denote by $v_{P}$ the normalized discrete valuation associated with $P$.

Let $\mathbb{P}_F $ denote the set of all places of $F$. The quotient $\mathcal{O}_{P}/P$ is called the \emph{residue class field} of $P$. The residue class map $\mathcal{O}_{P}\to \mathcal{O}_{P}/P$ induces a canonical embedding of $K$ into $\mathcal{O}_{P}/P$, and the extension degree $[\mathcal{O}_{P}/P:K]$ is called the \emph{degree} of $P$. A place of degree one is called a \emph{rational place} of $F/K$. These rational places will play a central role in our study of splitting behavior.

For a rational function field \( K(x) \), the finite places correspond bijectively to irreducible monic polynomials in $K[x]$. In particular, for $a\in K$, we denote by $P_{a}$ the rational place corresponding to $x-a$. The infinite place \( P_\infty \) corresponds to the valuation associated with $1/x$.

In this paper, we focus on the case $K=\mathbb{F}_{q}$. In this setting, there are exactly $q+1$ rational places in the rational function field $\mathbb{F}_{q}(x)$. Any finite extension of $\mathbb{F}_{q}(x)$ is called a \emph{global function field}, which provides a natural analogue of number fields over finite fields.

Let $F'/F$ be a finite separable extension of function fields. One of the fundamental phenomena in this context is the splitting of places: a place of \( F \) may decompose into several places in \( F' \), each carrying specific ramification and residue properties.

We say that a place \( P' \in \mathbb{P}_{F'} \) lies over a place \( P \in \mathbb{P}_F \), or that $P'$ is an extension of $P$, if \( P \subseteq P' \), and we write $P'\mid P$. For any place \( P' \in \mathbb{P}_{F'} \), there exists a unique place \( P \in \mathbb{P}_F \) such that \( P' \mid P \). Conversely, every place \( P \in \mathbb{P}_F \) has at least one, but only finitely many, places \( P' \in \mathbb{P}_{F'} \) lying over it.

The splitting behavior of a place is described by two fundamental invariants: the ramification index and the relative degree, defined as follows.

\begin{definition}
	Let \( P' \mid P \). The ramification index \( e(P' \mid P) \) is the integer such that \( v_{P'}(z) = e(P' \mid P) \cdot v_P(z) \) for all \( z \in F \). The relative degree \( f(P' \mid P) \) is \( [\mathcal{O}_{P'}/P' : \mathcal{O}_P/P] \). We say that \( P' \mid P \) is ramified if $e(P' \mid P)>1$, and unramified if $e(P' \mid P)=1$. We say that $P$ is ramified in $F'/F$ if there exists at least one place $P'\in \mathbb{P}_{F'}$ such that $P'\mid P$ and $e(P' \mid P)>1$.
\end{definition}

These invariants encode the arithmetic complexity of the extension and will be crucial in analyzing the behavior of rational places in our setting.

\begin{theorem}[Fundamental equality]
	Let $F'/F$ be a finite separable extension of function fields and let \( P \in \mathbb{P}_F \). If \( P_1, \ldots, P_m \) are all places of \( F' \) lying over \( P \), then
	\[
	\sum_{i=1}^m e(P_i \mid P) f(P_i \mid P) = [F' : F].
	\]
\end{theorem}

This equality is a cornerstone of the theory, expressing a precise balance between ramification and residue degrees in field extensions.

We say that a place \( P \) \emph{totally splits} in \( F'/F \) if there are exactly \( [F' : F] \) places lying over \( P \); in this case, by the fundamental equality, each of these places satisfies \( e = 1 \) and \( f = 1 \). 

An extension $P'\mid P$ is said to be \emph{tamely (resp., wildly) ramified} if $e(P'\mid P)>1$ and $\operatorname{char}(K)$ does not divide $e(P'\mid P)$ (resp., $\operatorname{char}(K)$ divides $e(P'\mid P)$). The extension $F'/F$ is called \emph{tame} if no place of $F$ is wildly ramified in $F'/F$.

The Hurwitz genus formula is a powerful tool for studying the global structure of function fields and their extensions, in particular for controlling the genus and ramification.

\begin{theorem}[Hurwitz genus formula]
	Let $F'/F$ be a finite separable extension of function fields having the same constant field $K$. Let $g$ (resp. $g'$) denote the genus of $F/K$ (resp. $F'/K$). Then
	\[
	2g'-2 \geq [F' : F] \cdot (2g-2) + \sum_{P \in \mathbb{P}_F} \sum_{P' \mid P} (e(P' \mid P)-1) \cdot \deg P'.
	\]
	Equality holds if and only if $F'/F$ is tame.
\end{theorem}

This formula will play a central role in deriving bounds and structural properties in the later sections, especially in the analysis of rational functions via their associated function field extensions.

\section{LRCs from good rational functions}\label{LRCs from good rational functions}

In this section, we introduce the notion of good rational functions and show how they naturally lead to constructions of optimal locally recoverable codes. This framework can be seen as a conceptual and structural extension of the classical Tamo--Barg construction.

Let $h(x)=\displaystyle\frac{f(x)}{g(x)}\in \mathbb{F}_q(x)$ be a nonconstant rational function, with $\gcd(f(x),g(x))=1$, and degree $\deg(h)=\max\{\deg(f), \deg(g)\}$. Consider the projective line $\mathbb{P}^{1}(\mathbb{F}_{q})=\mathbb{F}_{q}\cup \left \{ \infty  \right \}$. Then $h(x)$ naturally induces a map
\[
h : \mathbb{P}^1(\mathbb{F}_q) \to \mathbb{P}^1(\mathbb{F}_q),
\]
defined by evaluation.

For any point $P \in \mathbb{P}^1(\mathbb{F}_q)$, we define the fiber
\[
h^{-1}(P) \triangleq \{ Q \in \mathbb{P}^1(\mathbb{F}_q) \mid h(Q) = P \}.
\]
By basic properties of rational functions, one always has
\[
\# h^{-1}(P) \leq \deg(h).
\]
This inequality reflects the fact that $h$ behaves as a finite morphism of degree $\deg(h)$.

As a natural generalization of \emph{good polynomials}, we introduce the following notion of \emph{good rational function}, which will serve as the key ingredient in our construction.

\begin{definition}\label{good rational}
Let $h(x)\in \mathbb{F}_q(x)$ be a rational function, and let $l$ be a positive integer. Then $h(x)$ is said to be $(r,l)$-good if
\begin{itemize}
\item[(i)] $h(x)$ has degree $r+1$.
\item[(ii)] There exist $l$ subsets $A_{1},\ldots, A_{l}$ of $\mathbb{P}^{1}(\mathbb{F}_{q})$, such that 
\begin{itemize}
	\item $|A_{i}|=r+1$, and $A_{i}\cap A_{j}=\emptyset$, if $i\ne j$.
	\item $h(A_{i})=\{t_{i}\}$ for some $t_{i}\in\mathbb{P}^{1}(\mathbb{F}_{q}) $.
\end{itemize}
\end{itemize}
\end{definition}

This definition captures the essential feature required for locality: each subset $A_i$ forms a local repair group on which $h$ is constant.

We say that a rational function $h$ is \emph{good} if there exists a point $P\in \mathbb{P}^{1}(\mathbb{F}_{q})$ such that $\# h^{-1}(P)= \deg(h)$; otherwise, it is said to be \emph{not good}. This condition ensures the existence of fibers of maximal size, which is crucial for constructing large families of recovery sets.

Inspired by the Tamo--Barg framework, we now present a construction of LRCs based on $(r,l)$-good rational functions. The key idea is to exploit the structure of the fibers of $h$ to define local recovery groups.

Since an $(r,l)$-good rational function has degree $r+1\ge2$ and a fiber containing at least two rational points, it omits at least one value of $\mathbb{P}^{1}(\mathbb{F}_q)$. Replacing $h$ by $1/h$ if necessary, we may assume that an omitted value $b$ belongs to $\mathbb{F}_q$.

\begin{theorem}\label{main construction}
	Let $r,k,l$ be positive integers such that $r\mid k$ and $k/r\le l$.
	Let $h$ be an $(r,l)$-good rational function.
	After replacing $h$ by $1/h$ if necessary, write
	$h(x)=\displaystyle\frac{f(x)}{g(x)}\in\mathbb{F}_q(x)$
	with $\gcd(f,g)=1$, and choose
	$b\in\mathbb{F}_q$ such that
	$b\notin h(\mathbb{P}^{1}(\mathbb{F}_q))$.
	For $a\in\mathbb{F}_q^k$, written as $a_{ij}$ with
	$0\le i\le r-1$ and $0\le j\le k/r-1$, let
	\[
	h_a(x)=\sum_{i=0}^{r-1}\sum_{j=0}^{\frac{k}{r}-1}
	a_{ij}\left(\frac{g(x)}{f(x)-bg(x)}\right)^j x^i.
	\]
	Suppose that $h=t_i$ on $A_i$, where
	$t_1,\ldots,t_l$ are distinct and $|A_i|=r+1$
	$(1\le i\le l)$, and let $A=\bigcup_{i=1}^l A_i$.
	
	\begin{itemize}
		\item[(i)] If $\infty\notin A$, construct
		\[
		\mathcal{C}=\left\{\big(h_a(x):x\in A\big)
		\mid a\in\mathbb{F}_q^k\right\}.
		\]
		
		\item[(ii)] If $\infty\in A_e$ for some $1\le e\le l$,
		construct
		\[
		\mathcal{C}=\left\{\Big(\big(h_a(x):x\in A\setminus\{\infty\}\big),
		c_\infty\Big)\mid a\in\mathbb{F}_q^k\right\},
		\]
		where
		\[
		c_\infty=\sum_{j=0}^{\frac{k}{r}-1}a_{(r-1)j}
		\left(\frac{1}{t_e-b}\right)^j,
		\]
		with the convention $1/(\infty-b)=0$ and $0^0=1$.
	\end{itemize}
	
	Then $\mathcal{C}$ is an optimal $(n,k,r)$-LRC over
	$\mathbb{F}_q$, where $n=(r+1)l$.
\end{theorem}

\begin{proof}
	The encoding is linear. Since $b\notin h(\mathbb{P}^{1}(\mathbb{F}_{q}))$,
	$f-bg$ has no zero in $\mathbb{F}_{q}$. Moreover, $\deg(f-bg)=r+1$,
	otherwise $h(\infty)=b$.
	
	For $m=1,\ldots,l$, let $d_m=1/(t_m-b)$, where $d_m=0$ if
	$t_m=\infty$. On the finite points of $A_m$, we have
	\[
	h_a(x)=\sum_{i=0}^{r-1}
	\left(\sum_{j=0}^{\frac{k}{r}-1}a_{ij}d_m^j\right)x^i.
	\]
	The $d_m$ are pairwise distinct. Suppose that all the finite coordinates
	of a codeword are zero. Since $A_m$ contains at least $r$ finite points,
	the polynomial on the right, of degree at most $r-1$, must be zero. Thus
	\[
	\sum_{j=0}^{\frac{k}{r}-1}a_{ij}d_m^j=0,
	\qquad 0\le i\le r-1,\quad 1\le m\le l.
	\]
	For each fixed $i$, the left-hand side is a polynomial in $d_m$ of
	degree at most $k/r-1$. As $l\ge k/r$, we obtain $a_{ij}=0$ for all
	$i,j$. Hence the encoding is injective and $\dim\mathcal C=k$.
	
	Now write $h_a(x)$ over the common denominator
	$(f(x)-bg(x))^{\frac{k}{r}-1}$. For a nonzero message, the numerator is
	nonzero and has degree at most
	\[
	(r+1)\left(\frac{k}{r}-1\right)+r-1
	=k+\frac{k}{r}-2\le n-2.
	\]
	If $\infty\notin A$, the number of zero coordinates is therefore at
	most $k+\frac{k}{r}-2$.
	
	Suppose that $\infty\in A_e$. The coefficient of $x^{r+1}$ in $g$,
	divided by the leading coefficient of $f-bg$, equals
	$1/(t_e-b)$, with the convention $1/(\infty-b)=0$ and $0^0=1$. Consequently,
	the coefficient of $x^{k+\frac{k}{r}-2}$ in the numerator is
	$c_\infty$ multiplied by the $(\frac{k}{r}-1)$-st power of the leading
	coefficient of $f-bg$. If $c_\infty=0$, the degree of the numerator
	is at most $k+\frac{k}{r}-3$. Including the zero coordinate at
	$\infty$, there are again at most $k+\frac{k}{r}-2$ zero coordinates.
	The same bound is immediate if $c_\infty\ne0$. Hence, in either case,
	\[
	d\ge n-k-\frac{k}{r}+2.
	\]
	
	It remains to verify the locality. On each $A_m$ not containing
	$\infty$, the coordinates are evaluations of a polynomial of degree
	at most $r-1$, so any one is recovered from the other $r$ coordinates
	by interpolation. If $\infty\in A_e$, the coefficient of $x^{r-1}$
	of this polynomial is $c_\infty$. Thus $c_\infty$ is recovered from
	the $r$ finite evaluations. If one of the finite coordinates is erased,
	subtracting the known term $c_\infty x^{r-1}$ leaves a polynomial of
	degree at most $r-2$, which is determined by the remaining $r-1$
	finite evaluations. When $r=1$, the two coordinates on $A_e$ are
	simply equal. Therefore $\mathcal C$ has locality $r$. Finally, the
	Singleton-type bound gives
	\[
	d\le n-k-\left\lceil\frac{k}{r}\right\rceil+2
	=n-k-\frac{k}{r}+2,
	\]
	which proves the optimality of $\mathcal C$.
\end{proof}

	\section{Good rational functions and algebraic function field theory}\label{how good}

Under the construction framework developed in Theorem~\ref{main construction}, in order to obtain LRCs with a large range of possible lengths and dimensions, it is crucial to identify good rational functions for which the parameter $l$ is as large as possible.

For a nonconstant rational function $h$, define
\begin{equation}\label{inv}
l_h=\#\left\{P \in \mathbb{P}^1(\mathbb{F}_q)\mid\#h^{-1}(P)=\deg(h)\right\}.
\end{equation}
Thus a rational function of degree $r+1$ is $(r,l)$-good precisely when $1\le l\le l_h$. This observation naturally leads to a deeper investigation of the structure of these fibers.

To analyze the size of the set in \eqref{inv}, we embed the problem into the framework of algebraic function field theory, where powerful tools are available to study such questions. We begin by establishing the general setting.

Let $h(x)=\displaystyle\frac{f(x)}{g(x)}\in \mathbb{F}_q(x)$ be a rational function with $\gcd(f(x),g(x))=1$, and set $h(x)=t$. Then the extension $\mathbb{F}_q(x)/\mathbb{F}_q(t)$ is a function field extension of degree $\deg(h)$.

Throughout this section, we assume that $h$ is nonconstant and separable, equivalently,
\[
f'(X)g(X)-f(X)g'(X)\ne0.
\]
Thus $f(X)-tg(X)$ is separable over $\mathbb{F}_q(t)$,
and the extension $\mathbb{F}_q(x)/\mathbb{F}_q(t)$
is separable. This separability condition is essential for applying the standard machinery of algebraic function field theory.

We denote by $M_{h}$ the splitting field of the polynomial $f(X)-tg(X)$, that is, the Galois closure of the extension $\mathbb{F}_q(x)/\mathbb{F}_q(t)$. This field will play a central role in our analysis, as it allows us to interpret splitting phenomena in a Galois-theoretic framework.

We identify each point in $ \mathbb{P}^{1}(\mathbb{F}_{q})$ with a rational place of the rational function field, and denote by $P_{a}$ the place corresponding to $a\in \mathbb{P}^{1}(\mathbb{F}_{q})$.

The following fundamental observation establishes a precise correspondence between the fibers of $h$ and the splitting behavior of places. This connection will serve as the cornerstone of our approach.

\begin{proposition}\label{split com}
Let $h(x)$ be a nonconstant separable rational function,
and let $a,b\in\mathbb{P}^{1}(\mathbb{F}_{q})$, with $t=h(x)$.
Then $h(a)=b$ if and only if the place $P_a$ in
$\mathbb{F}_{q}(x)$ lies over the place $P_b$ in
$\mathbb{F}_{q}(t)$.

In particular, $\#h^{-1}(b)=\deg(h)$ if and only if
$P_b$ splits totally in the extension
$\mathbb{F}_q(x)/\mathbb{F}_q(t)$.
\end{proposition}

\begin{proof}
For finite $b$, the function $t-b$ is a local parameter
at $P_b$, while for $b=\infty$, one may take $1/t$.
Substituting $t=h(x)$, we see that the corresponding
function vanishes at $P_a$ precisely when $h(a)=b$,
where evaluation at infinity is understood in the
projective sense. This proves the first assertion.

The second assertion follows from the fundamental
equality, since $P_b$ has $\deg(h)$ distinct rational
places lying above it precisely when it splits totally.
\end{proof}

	By Proposition~\ref{split com}, the set \eqref{inv} can be interpreted as the set of totally split rational places in the extension $\mathbb{F}_q(x)/\mathbb{F}_q(h(x))$. We denote this set by $T^1_{\mathrm{split}}(h)$, so that $l_h=\#T^1_{\mathrm{split}}(h)$.

Therefore, our main objective is to estimate the quantity $\#T^1_{\mathrm{split}}(h)$, which plays a central role in determining how good a rational function is in the sense of Definition~\ref{good rational}.

A key idea in our approach is to transfer the problem to a Galois setting, where the structure of the extension becomes more tractable. Indeed, by the following standard result (see \cite{Stichtenoth2009}), it suffices to study the splitting behavior in the Galois closure $M_{h}/\mathbb{F}_q(t)$.

\begin{lemma}\label{M}
Let $h(x)$ be a rational function, and let $t=h(x)$. Let $M_{h}$ be the Galois closure of the extension $\mathbb{F}_q(x)/\mathbb{F}_q(t)$. Then for any place $P$ in $\mathbb{F}_{q}(t)$, we have
\begin{itemize}
\item[(i)] $P$ splits totally in the extension $\mathbb{F}_q(x)/\mathbb{F}_q(t)$ if and only if it splits totally in the extension $M_{h}/\mathbb{F}_q(t)$. 
\item[(ii)] $P$ is ramified in the extension $\mathbb{F}_q(x)/\mathbb{F}_q(t)$ if and only if it is ramified in the extension $M_{h}/\mathbb{F}_q(t)$. 
\end{itemize}
\end{lemma}

This lemma shows that passing to the Galois closure does not alter the essential splitting and ramification properties. Consequently, the study of $\#T^1_{\mathrm{split}}(h)$ can be carried out in the more structured and symmetric setting of a Galois extension.

We now establish a necessary condition for a rational function to be good, expressed in terms of the constant field of the Galois closure. This condition will be assumed throughout the remainder of the paper.

\begin{proposition}\label{not good}
If a rational function $h$ is good, then the constant field of $M_{h}$ is $\mathbb{F}_{q}$.
\end{proposition}

\begin{proof}
Assume that $h$ is good. Then there exists a rational place $P\in \mathbb{P}^{1}(\mathbb{F}_{q})$ such that $\# h^{-1}(P)= \deg(h)$. 

By Proposition~\ref{split com} and Lemma~\ref{M}, the place $P$ splits totally in the extension $M_{h}/\mathbb{F}_{q}(t)$.

Since $\mathbb{F}_{q}$ is the constant field of $\mathbb{F}_{q}(t)$, it is naturally contained in $M_{h}$. Let $P'$ be any place of $M_{h}$ lying over $P$. Because $P$ splits totally, we have $f(P'\mid P)=1$, that is,
\begin{equation}
\left [ \mathcal{O}_{P'}/P':\mathcal{O}_{P}/P\right ]=1.
\end{equation}

It follows that the residue class fields satisfy
\[
\mathcal{O}_{P'}/P'\cong \mathcal{O}_{P}/P\cong \mathbb{F}_{q}.
\]

On the other hand, the constant field of $M_{h}$ is canonically embedded into each residue field $\mathcal{O}_{P'}/P'$. Therefore, it must be a subfield of $\mathbb{F}_{q}$, and hence coincides with $\mathbb{F}_{q}$.

This completes the proof.
\end{proof}
	
	We now further analyze the structure of the Galois function field extension $M_{h}/\mathbb{F}_q(t)$, which provides the appropriate framework for studying the splitting behavior of places.

Let $G_{h}=\mathrm{Gal}(M_{h}/\mathbb{F}_q(t))$ be the Galois group of the extension. Let $P$ be a place of $\mathbb{F}_q(t)$. For any extension of $P$ to $M_{h}$, the ramification index $e(P)$ and the relative degree $f(P)$ depend only on $P$ (and not on the chosen extension). Moreover, the Galois group $G_{h}$ acts transitively on the set of places of $M_{h}$ lying over $P$.

Let $P'$ be a place of $M_{h}$ lying over $P$. The stabilizer
\[
D(P'\mid P):=\{\sigma\in G_{h}\mid \sigma(P')=P'\}
\]
is called the \emph{decomposition group} of $P'$ over $P$. Its cardinality is given by $|D(P'\mid P)|=e(P)f(P)$.

Similarly, the subgroup
\[
I(P'\mid P):=\{\sigma\in G_{h}\mid\sigma(z)-z\in P', \ \forall z\in \mathcal{O}_{P'}\}
\]
is called the \emph{inertia group} of $P'$ over $P$, and has order $e(P)$. It is a normal subgroup of $D(P'\mid P)$.

The following classical result describes the relationship between these groups and the induced Galois action on residue fields.

\begin{theorem}[\cite{Stichtenoth2009}]
In the extension $M_{h}/\mathbb{F}_q(t)$, let $P'$ be a place lying over a place $P$ of $\mathbb{F}_q(t)$. Write $F_P=\mathcal{O}_P/P$ and $F'_{P'}=\mathcal{O}_{P'}/P'$. Then the residue class extension $F'_{P'}/F_P$ is Galois. Each automorphism $\sigma\in D(P'\mid P)$ induces an automorphism $\overline{\sigma}$ in $\mathrm{Gal}(F'_{P'}/F_{P})$ by setting
\[
\overline{\sigma}(z(P'))=\sigma(z)(P') \quad \text{for } z\in \mathcal{O}_{P'}.
\]
The mapping
\[
D(P'\mid P)\longrightarrow \mathrm{Gal}(F'_{P'}/F_{P}), \quad \sigma \longmapsto \overline{\sigma}
\]
is a surjective homomorphism whose kernel is the inertia group $I(P'\mid P)$. In particular, one has the isomorphism
\[
D(P'\mid P)/I(P'\mid P)\cong \mathrm{Gal}(F'_{P'}/F_{P}).
\]
\end{theorem}

We now specialize to the case where $P$ is a rational place of $\mathbb{F}_{q}(t)$. In this situation, one has $F_{P}=\mathbb{F}_{q}$, and the residue field extension $F'_{P'}/\mathbb{F}_{q}$ is finite and Galois.

For an automorphism $\sigma\in G_{h}$ and a rational place $P$ of $\mathbb{F}_{q}(t)$, we say that $\sigma\in D(P'\mid P)$ is a \emph{Frobenius} for $P$ if the induced automorphism $\overline{\sigma}$ acts as
\[
\overline{\sigma}(z)=z^q\qquad(z\in F'_{P'}),
\]
which generates the cyclic group $\mathrm{Gal}(F'_{P'}/\mathbb{F}_{q})$.

This notion provides a natural link between the arithmetic of the base field and the Galois action on places. In particular, it is easy to verify that $f(P)=1$ if and only if the identity automorphism acts as a Frobenius for $P$.

We are now in a position to derive a formula for $\# T^1_{\mathrm{split}}(h)$, which is the central quantity of interest. The proof relies on analytic tools from the theory of global function fields, namely the \emph{zeta function} and \emph{Artin $L$-series}, which encode the distribution of places in Galois extensions.

Denote by $R_{h}$ (resp., $R_{h}^{1}$) the set of all ramified places (resp., rational ramified places) of $\mathbb{F}_q(t)$ in the extension $M_{h}/\mathbb{F}_{q}(t)$.

\begin{theorem}\label{split}
	Let $h\in\mathbb{F}_q(x)$ be a rational function, and let $M_{h}$ be the Galois closure of the extension $\mathbb{F}_q(x)/ \mathbb{F}_q(t)$. Assume that the constant field of $M_{h}$ is $\mathbb{F}_q$. Denote by $G_{h}$ the Galois group of $M_{h}/ \mathbb{F}_q(t)$.
	
	Let $\mathbb{P}_{M_h}^{1}$ be the set of rational places in $M_{h}$, and let
	$R_{h}^{1}$ be the set of rational places of $\mathbb{F}_q(t)$
	ramified in the extension $M_h/\mathbb{F}_q(t)$. Then one has
	\begin{equation}
		\begin{aligned}
			\# T^1_{\mathrm{split}}(h)&=\frac{\#\mathbb{P}_{M_h}^{1}}{\#G_h}-\sum_{\substack{P \in R_h^1\\ f(P)=1}} \frac{1}{e(P)},
		\end{aligned}
	\end{equation}
	where $e(P)$ and $f(P)$ denote the ramification index and the relative degree of the place $P$, respectively.
\end{theorem}

\begin{proof}
	
	We analyze the Galois extension $M_{h}/\mathbb{F}_{q}(t)$ by exploiting the deep connection between its zeta function, that of the rational function field $\mathbb{F}_q(t)$, and the Artin $L$-series associated with irreducible characters of the Galois group $G_{h}$. This connection provides a powerful tool to count rational places.

	More precisely, one has the following factorization:
	\begin{equation}\label{product}
		\zeta_{M_{h}}(s)=\zeta_{\mathbb{F}_q(t)}(s)\prod_{\substack{\chi\ne \chi_{0} \\ \chi\in \operatorname{Irr}(G_{h})}}L(s,\chi)^{\chi(1)},
	\end{equation}
	where $\chi$ runs through all non-trivial irreducible characters of $G_h$. Here, $\operatorname{Irr}(G_h)$ denotes the set of irreducible characters of $G_h$, and $\chi_0$ is the trivial character. We refer to \cite{Rosen 2013} for further details on this decomposition.

	Taking logarithms on both sides of \eqref{product}, we obtain
	\begin{equation}\label{log}
		\log \zeta_{M_{h}}(s)=\log \zeta_{\mathbb{F}_q(t)}(s) +\sum_{\substack{\chi\ne \chi_{0} \\ \chi\in \operatorname{Irr}(G_{h})}} \chi(1)\log L(s,\chi).
	\end{equation}
	
	Let $u=q^{-s}$. For a rational place 
	$P\in\mathbb{P}^{1}(\mathbb{F}_{q})$, define
	\[
	\chi(P)=
	\begin{cases}
		\chi(\tau_P), & P \text{ is unramified}, \\[1ex]
		\displaystyle\frac{1}{e(P)}
		\sum_{\omega\in I(P'\mid P)}\chi(\omega\tau_P),
		& P \text{ is ramified},
	\end{cases}
	\]
	where $P'$ is an extension of $P$ to $M_h$, and $\tau_P$ is a
	Frobenius element in $D(P'\mid P)$.
	
	From the Euler product of the zeta function, we have
	\[
	\log\zeta_{M_h}(s)
	=
	\#\mathbb{P}_{M_h}^{1}\,u+O(u^2),
	\]
	and
	\[
	\log\zeta_{\mathbb{F}_q(t)}(s)
	=
	(q+1)u+O(u^2).
	\]
	Moreover, from the Euler product defining the Artin $L$-series,
	\[
	\log L(s,\chi)
	=
	\sum_{P\in\mathbb{P}^{1}(\mathbb{F}_{q})}
	\chi(P)\,u+O(u^2).
	\]
	
	Therefore, comparing the coefficients of $u$ in \eqref{log}, we obtain
	\begin{equation}\label{uuuu}
		\#\mathbb{P}_{M_h}^{1}
		=
		(q+1)
		+
		\sum_{\substack{\chi\ne\chi_0\\ \chi\in \operatorname{Irr}(G_h)}}
		\chi(1)
		\sum_{P\in\mathbb{P}^{1}(\mathbb{F}_{q})}\chi(P).
	\end{equation}
	
	Since
	\[
	q+1
	=
	\sum_{P\in\mathbb{P}^{1}(\mathbb{F}_{q})}
	\chi_0(1)\chi_0(P),
	\]
	equation~\eqref{uuuu} can be rewritten as
	\[
	\#\mathbb{P}_{M_h}^{1}
	=
	\sum_{P\in\mathbb{P}^{1}(\mathbb{F}_{q})}
	\sum_{\chi\in \operatorname{Irr}(G_h)}
	\chi(1)\chi(P).
	\]

	We now analyze separately the contributions of unramified and ramified places.

	\noindent
	\textbf{Unramified places.}
	For an unramified place $P \in  \mathbb{P}^{1}(\mathbb{F}_{q})$, we use the character identity associated with the regular representation:
	\begin{equation}
		\sum_{\chi\in \operatorname{Irr}(G_{h}) }\chi(1)\chi(P)
		=
		\sum_{\chi\in \operatorname{Irr}(G_{h}) }\chi(1)\chi(\tau_P)
		=
		\chi_{\mathrm{reg}}(\tau_P)
		=
		\left\{
		\begin{array}{cl}
			0, &  \tau_P \neq 1,  \\
			\#G_{h}, & \tau_P = 1.
		\end{array}
		\right.
	\end{equation}
	
	Recall that $P$ splits totally if and only if it is unramified and $\tau_P = 1$. Hence, the contribution of totally split places yields
	\begin{equation}\label{half}
		\begin{aligned}
			\#\mathbb{P}_{M_h}^{1}
			=
			\#T^1_{\mathrm{split}}(h)\cdot \#G_{h}
			+
			\sum_{P\in R_{h}^{1}} \sum_{\chi\in \operatorname{Irr}(G_{h})} \chi (1) \chi(P).
		\end{aligned}
	\end{equation}

	\noindent
	\textbf{Ramified places.}
	For a ramified place $P\in R_{h}^1$, one obtains
	\begin{equation}\label{hhalf}
		\sum_{\chi\in \operatorname{Irr}(G_{h}) }\chi(1)\chi(P)
		=
		\frac{1}{e(P)}\sum_{\omega \in I(P'\mid P)} \sum_{\chi\in \operatorname{Irr}(G_{h}) }\chi (1) \chi(\omega \tau_P)
		=
		\left\{
		\begin{array}{cl}
			0, &  f(P) \neq 1,  \\
			\displaystyle\frac{\#G_{h}}{e(P)}, & f(P) = 1.
		\end{array}
		\right.
	\end{equation}

	Finally, combining \eqref{half} and \eqref{hhalf}, we obtain
	\begin{equation}\label{fi}
		\begin{aligned}
			\# T^1_{\mathrm{split}}(h)
			=
			\frac{\#\mathbb{P}_{M_h}^{1}}{\#G_h}
			-
			\sum_{\substack{P \in R_h^1\\ f(P)=1}} \frac{1}{e(P)}.
		\end{aligned}
	\end{equation}
	
	This completes the proof.
\end{proof}

We now combine Theorem~\ref{split} with the classical Hasse--Weil bound to derive explicit estimates for $\#T^1_{\mathrm{split}}(h)$. This provides a quantitative understanding of how the arithmetic of the function field $M_h$ influences the quality of the rational function $h$.

We first recall the well-known Hasse--Weil bound (see \cite{Stichtenoth2009}).

\begin{lemma}[Hasse--Weil bound]\label{hasse}
Let $F/\mathbb{F}_{q}$ be a function field with full constant field $\mathbb{F}_{q}$ and genus $g$, and let $N_{1}(F)$ be the number of rational places. Then 
\begin{equation}
|N_{1}(F)-(q+1)| \leq 2g\sqrt{q}.
\end{equation}  
\end{lemma}

Applying this bound to the function field $M_h$, and combining it with Theorem~\ref{split}, we immediately obtain the following estimate.

\begin{corollary}\label{estimate}
Let $M_{h}$ be the Galois closure of $\mathbb{F}_q(x)/\mathbb{F}_q(t)$ with Galois group $G_{h}$. If the constant field of $M_{h}$ is $\mathbb{F}_q$, then
\begin{equation}\label{es}
\begin{gathered}
\left \lceil   \frac{q+1-2 g_h \sqrt{q}}{\# G_{h}}-\frac{\# R^1_h }{2}\right \rceil
\leq \# T^1_{\mathrm{split}}(h) \leq \left \lfloor   \frac{q+1+2 g_h \sqrt{q}}{\# G_{h}} \right \rfloor,
\end{gathered}
\end{equation}
where $g_h$ is the genus of $M_{h}$, and $R^1_h$ is the set of all rational ramified places of $\mathbb{F}_q(t)$ in the extension $M_{h}/ \mathbb{F}_q(t)$.
\end{corollary}

\begin{proof}
By the Hasse--Weil bound (Lemma~\ref{hasse}), we have
\[
q+1-2g_h\sqrt{q} \le \#\mathbb{P}_{M_h}^{1} \le q+1+2g_h\sqrt{q}.
\]

On the other hand, for any ramified place $P\in R_{h}^{1}$, one has $e(P)\ge 2$, and hence
\[
\frac{1}{e(P)} \le \frac{1}{2}.
\]

Substituting these bounds into Equation~\eqref{fi} yields the desired inequalities in \eqref{es}.
\end{proof}

In Corollary~\ref{estimate}, the lower bound of $\# T^1_{\mathrm{split}}(h)$ depends on the parameters $g_{h}$, $\# R^1_h$, and $\# G_{h}$. Therefore, in the selection of good rational functions, it is desirable to identify those for which this lower bound is as large as possible.

To achieve this, we derive upper bounds for $g_{h}$ and $\# R^1_h$ by employing the Riemann--Hurwitz formula (see \cite{Stichtenoth2009}). A key step in this analysis is to understand the splitting behavior of the place at infinity $P_{\infty}$ in the extension $\mathbb{F}_{q}(x)/ \mathbb{F}_{q}(t)$. Although this can be obtained directly, we state it explicitly for completeness.

\begin{proposition}\label{infi}
Let $h(x)=\displaystyle\frac{f(x)}{g(x)}$ be a rational function with $f(x)$ and $g(x)$ coprime. Write $g(x)=c\prod_i g_i(x)^{a_i}$, where $c\in\mathbb{F}_q^*$ and the $g_i(x)$ are distinct monic irreducible polynomials. Denote by $P_{g_{i}}$ the place associated with $g_{i}(x)$ in $\mathbb{F}_{q}(x)$, and let $P_{\infty}^{t}$ (resp., $P_{\infty}^{x}$) be the place at infinity in $\mathbb{F}_{q}(t)$ (resp., $\mathbb{F}_{q}(x)$). 

Let $h(x)=t$. In the extension $\mathbb{F}_{q}(x)/\mathbb{F}_{q}(t)$, the places of $\mathbb{F}_{q}(x)$ lying over $P_{\infty}^{t}$, together with their ramification indices and relative degrees, are given as follows:

\begin{itemize}
\item[(i)] If $\deg(f)\le \deg(g)$, then the places lying over $P_{\infty}^{t}$ are the $P_{g_{i}}$, with $e(P_{g_{i}}\mid P_{\infty}^{t})=a_{i}$ and $f(P_{g_{i}}\mid P_{\infty}^{t})=\deg(g_{i})$.
\item[(ii)] If $\deg(f)>\deg(g)$, then the places lying over $P_{\infty}^{t}$ are the $P_{g_{i}}$ and $P_{\infty}^{x}$, with $e(P_{g_{i}}\mid P_{\infty}^{t})=a_{i}$, $f(P_{g_{i}}\mid P_{\infty}^{t})=\deg(g_{i})$, $e(P_{\infty}^{x}\mid P_{\infty}^{t})=\deg(f)-\deg(g)$, and $f(P_{\infty}^{x}\mid P_{\infty}^{t})=1$. 
\end{itemize} 
\end{proposition}

We now give an upper bound for the number $\#R_h$ of ramified places in terms of $\deg(h)$ and the poles of $h$, and establish the genus bound $g_h\le \frac{\deg(h)-3}{2}\#G_h+1$. Both bounds are sharp and will be used to estimate $\#T^1_{\mathrm{split}}(h)$.

\begin{theorem}\label{21e}
Let $h(x)=\displaystyle\frac{f(x)}{g(x)}$ be a nonconstant separable rational
function with $f(x)$ and $g(x)$ coprime. Let $t=h(x)$, and let $M_h$
be the Galois closure of $\mathbb{F}_q(x)/\mathbb{F}_q(t)$. Denote by
$R_h$ and $R_h^1$ the sets of ramified places and ramified rational
places of $\mathbb{F}_q(t)$ in this extension, respectively. Let
$g_h$ be the genus of $M_h$, and let $G_h$ be its Galois group over
$\mathbb{F}_q(t)$. Then
\begin{itemize}
\item[(i)]
\[
\#R_h^1\le \#R_h\le
\min\left\{2\deg(h)-2,\,
\deg(h)+\sum_i\deg(g_i)+\delta_g^f-1\right\},
\]
where the $g_i$ are the distinct monic irreducible factors of
$g(x)$, and
\[
\delta_g^f=
\begin{cases}
1,&\deg(f)>\deg(g),\\
0,&\deg(f)\le\deg(g).
\end{cases}
\]

\item[(ii)] Suppose that $\operatorname{char}(\mathbb{F}_q)\nmid\#G_h$
and that the constant field of $M_h$ is $\mathbb{F}_q$. Then
\[
g_h\le \frac{\deg(h)-3}{2}\#G_h+1.
\]
\end{itemize}
\end{theorem}

\begin{proof}
By Lemma~\ref{M}, the ramified places of the base field are the same
in $\mathbb{F}_q(x)/\mathbb{F}_q(t)$ and its Galois closure. For each
place $P$ of $\mathbb{F}_q(t)$, write
\[
\delta_P=\sum_{P'\mid P}(e(P'\mid P)-1)\deg(P'),
\]
where $P'$ runs over the places of $\mathbb{F}_q(x)$ above $P$.
The Riemann--Hurwitz formula gives
\begin{equation}\label{H1}
\sum_{P\in R_h}\delta_P\le 2\deg(h)-2.
\end{equation}
Since $\delta_P\ge\deg(P)$ for every ramified place $P$, we obtain
$\#R_h\le 2\deg(h)-2$.

By Proposition~\ref{infi}, the contribution at infinity is
\begin{equation}\label{H2}
\begin{aligned}
\delta_{P_\infty^t}
&=\sum_{P'\mid P_\infty^t}
(e(P'\mid P_\infty^t)-1)f(P'\mid P_\infty^t)\\
&=\deg(h)-\sum_i\deg(g_i)-\delta_g^f.
\end{aligned}
\end{equation}
Consequently,
\[
\begin{aligned}
\#R_h
&\le \sum_{P\in R_h\setminus\{P_\infty^t\}}\delta_P+1\\
&\le \deg(h)+\sum_i\deg(g_i)+\delta_g^f-1.
\end{aligned}
\]
This proves (i).

We now prove (ii). The extension $M_h/\mathbb{F}_q(t)$ is tame,
so its inertia groups are cyclic. In the natural faithful action
of $G_h$ on the roots of $f(X)-tg(X)$, the cycle lengths of an
inertia generator above $P$ are the indices $e(P'\mid P)$, each
repeated $f(P'\mid P)$ times. Thus, if $\delta_P=\deg(P)$, this
generator is a transposition and $e(P)=2$. Otherwise,
$\delta_P\ge 2\deg(P)$ for every ramified $P$. In either case,
\begin{equation}\label{inertia-contribution}
\left(1-\frac{1}{e(P)}\right)\deg(P)
\le \frac{\delta_P}{2}.
\end{equation}
Here, as before, $e(P)$ denotes the ramification index in the
Galois closure.

The lower extension is also tame, so equality holds in
\eqref{H1}. Applying Riemann--Hurwitz to the Galois closure and
using \eqref{inertia-contribution}, we obtain
\begin{equation}\label{ii1}
\begin{aligned}
2g_h-2
&=\#G_h\left(-2+
\sum_{P\in R_h}\left(1-\frac{1}{e(P)}\right)\deg(P)\right)\\
&\le \#G_h\left(-2+\frac12\sum_{P\in R_h}\delta_P\right)\\
&=(\deg(h)-3)\#G_h.
\end{aligned}
\end{equation}
This proves (ii).

Keeping the contribution at $P_\infty^t$ in the first line of
\eqref{ii1} and applying \eqref{inertia-contribution} to the finite
ramified places gives, by \eqref{H1} and \eqref{H2},
\begin{equation}\label{ii2}
g_h\le
\frac{\#G_h}{4}\left(
\deg(h)+\sum_i\deg(g_i)+\delta_g^f-4
-\frac{2}{e(P_\infty^t)}\right)+1,
\end{equation}
where $e(P_\infty^t)=1$ when $P_\infty^t$ is unramified.
\end{proof}

\begin{remark}
The bounds in Theorem~\ref{21e} are sharp.
For example, over $\mathbb{F}_{17}$, we consider
\[
h(x)=\frac{x(x-1)(x-2)}{x-3}.
\]
A direct calculation gives
\[
h'(x)=\frac{2(x-5)(x-6)(x-12)}{(x-3)^2}.
\]
Thus the finite ramification points are $x=5,6,12$, each
with ramification index $2$. Their images are
\[
h(5)=13,\qquad h(6)=6,\qquad h(12)=5.
\]
Moreover, $h$ has poles of orders $1$ and $2$ at $x=3$
and $x=\infty$, respectively. Hence the point at infinity
is also ramified with index $2$.
By the Riemann--Hurwitz formula, these account for all
the ramification. Therefore,
\[
R_h=R_h^1=\{P_5,P_6,P_{13},P_\infty^t\},
\]
and each of these places has ramification index $2$
in the Galois closure.

Since $h^{-1}(0)=\{0,1,2\}$, the constant field of $M_h$
is $\mathbb{F}_{17}$ by Proposition~\ref{not good}.
Furthermore, $G_h$ is a transitive subgroup of $S_3$
containing a transposition, thus $G_h\cong S_3$.
Applying the Riemann--Hurwitz formula to
$M_h/\mathbb{F}_{17}(t)$, we obtain
\[
2g_h-2
=6\left(-2+4\left(1-\frac12\right)\right)=0.
\]
Consequently,
\[
g_h=1=\frac{\deg(h)-3}{2}\#G_h+1.
\]
On the other hand,
\[
\#R_h^1=\#R_h=4
=\min\left\{2\deg(h)-2,\,
\deg(h)+\sum_i\deg(g_i)+\delta_g^f-1\right\}.
\]
Thus equality holds in both bounds of Theorem~\ref{21e}.
\end{remark}

Combining Corollary~\ref{estimate} and Theorem~\ref{21e}, we obtain
an explicit lower bound on the number of full rational fibers.

\begin{proposition}\label{pregene}
Let $h(x)=\displaystyle\frac{f(x)}{g(x)}$ be a rational function
with $f(x)$ and $g(x)$ coprime, and let $\deg(h)=r+1$. Let $M_h$
be the Galois closure of $\mathbb{F}_q(x)/\mathbb{F}_q(h(x))$.
Suppose that the constant field of $M_h$ is $\mathbb{F}_q$ and that
$\gcd(q,(r+1)!)=1$. Then $l=\#T^1_{\mathrm{split}}(h)$ satisfies
\[
l\ge \frac{q-2\sqrt q+1}{(r+1)!}
-(r-2)\sqrt q
-\frac{r+1+\sum_i\deg(g_i)}{2},
\]
where the $g_i$ are the distinct monic irreducible factors of $g(x)$.
\end{proposition}

\begin{proof}
The Galois group of $f(X)-tg(X)$ is a subgroup of $S_{r+1}$, so
$\#G_h\mid(r+1)!$. The condition $\gcd(q,(r+1)!)=1$ therefore
implies $\operatorname{char}(\mathbb{F}_q)\nmid\#G_h$.
By Corollary~\ref{estimate} and Theorem~\ref{21e},
\[
\begin{aligned}
l
&\ge \frac{q+1-2g_h\sqrt q}{\#G_h}-\frac{\#R_h^1}{2}\\
&\ge \frac{q-2\sqrt q+1}{\#G_h}
-(r-2)\sqrt q
-\frac{r+1+\sum_i\deg(g_i)}{2}.
\end{aligned}
\]
The result follows from $\#G_h\le(r+1)!$ and
$q-2\sqrt q+1=(\sqrt q-1)^2\ge0$.
\end{proof}

For two
nonconstant rational functions $u,v\in\mathbb{F}_q(x)$, one has
$\mathbb{F}_q(u(x))=\mathbb{F}_q(v(x))$ if and only if
$u=\phi\circ v$, where
\[
\phi(x)=\frac{ax+b}{cx+d},\qquad
 a,b,c,d\in\mathbb{F}_q,\quad ad-bc\ne0.
\]
In this case, we say that $u$ and $v$ are \emph{equivalent}.
They define the same intermediate field, and $\phi$ permutes the
values in $\mathbb{P}^1(\mathbb{F}_q)$. Hence equivalent functions
have the same number of full rational fibers. In the search for
good rational functions, it therefore suffices to consider
representatives up to this equivalence.

The following construction gives good rational functions that
are not equivalent to polynomials.

\begin{proposition}\label{linear-denominator-construction}
Let $r$ be a positive integer satisfying $\gcd(q,(r+1)!)=1$, and
let $S\subseteq\mathbb{F}_q$ with $\#S=r+1$. Set
\[
h(x)=\frac{\prod_{s\in S}(x-s)}{x-a},
\qquad a\in\mathbb{F}_q\setminus S.
\]
Then $h$ is $(r,l)$-good for $l=\#T^1_{\mathrm{split}}(h)$, and
\[
l\ge \frac{q-2\sqrt q+1}{(r+1)!}
-\frac{(r-2)(r+1)}{2r}\sqrt q
-\frac{r^2+r+2}{2r}.
\]
If $r=1$, then $l\ge(q-1)/2$.
\end{proposition}

\begin{proof}
The fiber above zero is exactly $S$, so $h$ is good. By
Proposition~\ref{not good}, the constant field of $M_h$ is
$\mathbb{F}_q$.

As in the proof of Proposition~\ref{pregene}, the Galois closure
is tame. The poles of $h$ are $a$ and infinity, of orders $1$ and $r$,
respectively. Since the inertia group at $P_\infty^t$ is cyclic
and a generator has cycle lengths $1$ and $r$, we have
$e(P_\infty^t)=r$. Applying \eqref{ii2} gives
\[
g_h\le \frac{(r-2)(r+1)}{4r}\#G_h+1.
\]
Moreover, \eqref{H1} and \eqref{H2} imply that there are at most
$r+1$ finite rational ramified places. When $r\ge2$, Equation~\eqref{fi}
and the Hasse--Weil bound therefore give
\[
\begin{aligned}
l
&\ge \frac{q+1-2g_h\sqrt q}{\#G_h}
-\frac{r+1}{2}-\frac1r\\
&\ge \frac{q-2\sqrt q+1}{(r+1)!}
-\frac{(r-2)(r+1)}{2r}\sqrt q
-\frac{r^2+r+2}{2r}.
\end{aligned}
\]
When $r=1$, the extension is quadratic, $\#G_h=2$, and $g_h=0$.
Theorem~\ref{21e} gives $\#R_h^1\le2$, thus
Corollary~\ref{estimate} yields $l\ge(q-1)/2$.
\end{proof}

\begin{remark}
The rational function in  Proposition~\ref{linear-denominator-construction} is not
equivalent to any polynomial. Indeed, suppose that $h=\phi\circ u$
for a polynomial $u$ and a linear fractional transformation
$\phi$. Since $u(\infty)=h(\infty)=\infty$, we have
$\phi(\infty)=\infty$. Thus $\phi$ is affine and $h$ is a
polynomial, contradicting the pole at $a$.
\end{remark}

\section{The advantage of the construction framework}\label{Gal case}
Taking all full rational fibers of a good rational function $h$ of degree $r+1$, our construction produces optimal LRCs of length $(r+1)l_h$, where $l_h=\#T^1_{\mathrm{split}}(h)$. Corollary~\ref{estimate} suggests looking for rational functions whose Galois closures have small degree and genus. The smallest possible value of $\#G_h$ is $\deg(h)$, attained precisely when $\mathbb{F}_q(x)/\mathbb{F}_q(h(x))$ is Galois. We first  propose an explicit characterization of $\# T^1_{\mathrm{split}}(h)$ from the perspective of group actions, and some examples of Galois extensions are also presented. We further give a non-Galois construction using a quintic rational function with a small Galois closure. As shown by the comparisons with polynomials of the same degree, these examples  demonstrate the advantage of the rational-function framework in Theorem~\ref{main construction} in constructing optimal LRCs.

 \subsection{Characterization of Galois extensions}
 
We now exploit the Galois structure of the extension to obtain a precise description of the splitting behavior via group actions. Consider the natural action of $G_{h}$ on the set of rational places of $\mathbb{F}_{q}(x)$. We classify the orbits according to their size: an orbit is said to be \emph{long} if its size is $\#G_{h}=\deg(h)$, and \emph{short} otherwise.

The following result shows that short orbits are extremely limited in number.

\begin{theorem}\label{3 short}
Let $h(x)$ be a rational function such that the extension $\mathbb{F}_{q}(x)/\mathbb{F}_{q}(t)$ is Galois, where $t=h(x)$. Let $G_{h}$ be the Galois group. Consider its natural action on the rational places of $\mathbb{F}_{q}(x)$. Then 
\begin{itemize}
\item[(i)] There are at most three short orbits.
\item[(ii)] The total number of orbits is $\left \lceil \frac{q+1}{\deg(h)}  \right \rceil$ when there are $0$ or $1$ short orbits, and $\left \lceil \frac{q+1}{\deg(h)}  \right \rceil+1$ when there are $2$ or $3$ short orbits.
\end{itemize}
\end{theorem}

\begin{proof}
The action of $G_{h}$ partitions the $q+1$ rational places of $\mathbb{F}_{q}(x)$ into $w$ disjoint orbits of sizes $l_{1}\le \cdots \le l_{w}$.

Two rational places lie in the same orbit if and only if they lie over the same rational place of $\mathbb{F}_{q}(t)$. Hence, each orbit corresponds to a rational place $P$ of $\mathbb{F}_{q}(t)$, and its size is given by
\[
l_i = \frac{\deg(h)}{e(P)},
\]
where $e(P)$ is the ramification index.

Let the first $k$ orbits be short. Then $e(P)>1$ for these orbits. Applying the Riemann--Hurwitz formula, we obtain
\begin{equation}\label{2d-2}
\begin{aligned}
2\deg(h)-2
&\ge\sum_{\substack{P\in R_h^1\\f(P)=1}}
\frac{\deg(h)}{e(P)}(e(P)-1)\\
&=\sum_{i=1}^{k}\bigl(\deg(h)-l_i\bigr)\\
&\ge k\frac{\deg(h)}{2}.
\end{aligned}
\end{equation}

This implies $k\le 3$, proving (i).

We now turn to (ii). Clearly, the number of orbits satisfies
\[
w \ge \left \lceil \frac{q+1}{\deg(h)} \right \rceil.
\]

\begin{itemize}
\item If $w=1$, then $q+1=l_{1}\le \deg(h)$, hence $\left \lceil \frac{q+1}{\deg(h)} \right \rceil=1$, and there is at most one short orbit.

\item If $w=2$, then $q+1=l_{1}+l_{2}$, and one checks that $w=2=\left \lceil \frac{q+1}{\deg(h)} \right \rceil+1$ if and only if both orbits are short.
\end{itemize}

Now assume $w\ge 3$. Since the remaining orbits have size $\deg(h)$, inequality~\eqref{2d-2} gives
\[
l_{1}+l_{2}+l_{3}\ge \deg(h)+2.
\]

Therefore,
\begin{equation}\label{os}
\begin{aligned}
\frac{q+1}{\deg(h)}
=
\frac{l_{1}+l_{2}+l_{3}+\deg(h)(w-3)}{\deg(h)}
\ge
w-2+\frac{2}{\deg(h)}.
\end{aligned}
\end{equation}

This yields
\[
\left \lceil \frac{q+1}{\deg(h)} \right \rceil \le w \le \left \lceil \frac{q+1}{\deg(h)} \right \rceil +1.
\]

Moreover, equality $w=\left \lceil \frac{q+1}{\deg(h)} \right \rceil+1$ holds if and only if
\[
l_{1}+l_{2}+l_{3}\le 2\deg(h),
\]
which is equivalent to having $2$ or $3$ short orbits.

This completes the proof.
\end{proof}

We now deduce an explicit formula for $\# T^1_{\mathrm{split}}(h)$.

\begin{corollary}\label{good 3}
Let $h(x)$ be a rational function of degree $r+1$ such that the extension $\mathbb{F}_{q}(x)/\mathbb{F}_{q}(t)$, where $t=h(x)$, is Galois. Let $V_{s}$ denote the number of short orbits under the action of $G_{h}$. Then
\[
\# T^1_{\mathrm{split}}(h) =
\begin{cases}
\dfrac{q+1}{r+1}, & \text{if } V_s = 0, \\[2ex]
\left\lceil \dfrac{q+1}{r+1} \right\rceil - 1, & \text{if } V_s = 1,2, \\[2ex]
\left\lceil \dfrac{q+1}{r+1} \right\rceil - 2, & \text{if } V_s = 3.
\end{cases}
\]
\end{corollary}

\begin{proof}
By definition, $\# T^1_{\mathrm{split}}(h)$ is equal to the number of long orbits. The result follows directly from Theorem~\ref{3 short}(ii).
\end{proof}

\begin{remark}\label{remark bound}
The upper bound of $\# T^1_{\mathrm{split}}(h)$ given in Corollary~\ref{estimate} is 
\[
\left \lfloor   \frac{q+1}{r+1} \right \rfloor.
\]
By Corollary~\ref{good 3}, if $r+1\nmid q+1$, then
\[
\left \lfloor \frac{q+1}{r+1} \right \rfloor
=
\left \lceil \frac{q+1}{r+1} \right \rceil -1.
\]
Hence, the upper bound is achieved whenever there are one or two short orbits.
\end{remark}

\subsection{Constructions and comparisons}

In this subsection, we consider the small localities $r=2,3,4$. For suitable finite fields, we construct rational functions of degree $r+1$ with more full rational fibers than any polynomial of the same degree. The cases $r=2,3$ arise from Galois extensions, while the case $r=4$ uses a non-Galois quintic rational function. The comparisons are with the original Tamo--Barg construction based on good polynomials. For $d\ge2$, write
\[
M_q(d)=\max_{\substack{f\in\mathbb{F}_q[X]\\\deg(f)=d}}\#T^1_{\mathrm{split}}(f).
\]

We first point out that good rational functions that induce Galois extensions can be constructed via subgroups of the automorphism group of the rational function field $\mathbb{F}_{q}(x)$, which is exactly the projective general linear group \[\mathrm{PGL}(2, q)=
\left\{\frac{ax+b}{cx+d}\mid a,b,c,d\in\mathbb{F}_q,\ ad-bc\ne 0\right\}.
\] Specifically, for a subgroup $H\le\mathrm{PGL}(2,q)$, let $\mathbb{F}_q(x)^H$ denote its fixed field. The extension $\mathbb{F}_q(x)/\mathbb{F}_q(x)^H$ is Galois, with
\[
\mathrm{Gal}(\mathbb{F}_q(x)/\mathbb{F}_q(x)^H)=H.
\]
By L\"uroth's theorem (see \cite{Stichtenoth2009}), there exists $h_H(x)\in\mathbb{F}_q(x)$ such that $\mathbb{F}_q(x)^H=\mathbb{F}_q(h_H(x))$. This function is unique up to left composition by an element of $\mathrm{PGL}(2,q)$, and its degree satisfies
\[
\deg (h_H(x))=[ \mathbb{F}_q(x):\mathbb{F}_q(h_H(x))]=\#H.
\]

To obtain explicit rational functions, we now consider
cyclic subgroups $H=\langle\phi\rangle$.
Jin, Ma, and Xing \cite[Lemma V.1]{Jin Ma Xing 2020}
recognized that the sum of the images of $x$ generates
$\mathbb{F}_q(x)^H$ when $H$ is contained in a cyclic
group of order $q+1$. The following proposition shows
that the same formula holds for any non-affine element
$\phi\in\mathrm{PGL}(2,q)$, without requiring its order to divide $q+1$.

\begin{proposition}\label{gen con}
Let $\phi\in\mathrm{PGL}(2,q)\setminus\mathrm{AGL}(1,q)$
be a non-affine transformation of order $r+1$.
Let $\phi^i$ denote its $i$-th iterate, with
$\phi^0(x)=x$, and put
\[
h(x)=\sum_{i=0}^{r}\phi^i(x).
\]
Then $\deg(h)=r+1$ and
\[
\mathbb{F}_q(x)^{\langle\phi\rangle}
=\mathbb{F}_q(h(x)).
\]
In particular, $\mathbb{F}_q(x)/\mathbb{F}_q(h(x))$
is a Galois extension.
\end{proposition}

\begin{proof}
Clearly, $h(\phi(x))=h(x)$, thus $h$ belongs to the
fixed field of $\langle\phi\rangle$. We claim that
no nonidentity power of $\phi$ fixes $\infty$.
Otherwise, for some $1\le j\le r$, $\phi^j$ would
be a nonidentity affine transformation. If $\phi^j$
is a translation, its only fixed point is $\infty$,
which must also be fixed by $\phi$, a contradiction.
Otherwise, $\phi^j$ has two fixed points, one of
which is $\infty$. Since $\phi$ commutes with
$\phi^j$ but does not fix $\infty$, it must
interchange these two points. Hence $\phi^2$ is
the identity, contradicting the fact that $\phi^j$
is a nonidentity affine transformation.

Thus $x,\phi(x),\ldots,\phi^r(x)$ have distinct
simple poles. Their sum therefore has exactly
$r+1$ simple poles, and $\deg(h)=r+1$. Since
\[
[\mathbb{F}_q(x):\mathbb{F}_q(h(x))]
=r+1=|\langle\phi\rangle|,
\]
we obtain
$\mathbb{F}_q(h(x))
=\mathbb{F}_q(x)^{\langle\phi\rangle}$,
as required.
\end{proof}

Based on the general construction in Proposition~\ref{gen con}, we now give explicit examples of degree $3$ and $4$ ($r=2,3$), for which the extension $\mathbb{F}_{q}(x)/\mathbb{F}_{q}(h(x))$ is Galois. To illustrate the advantage of our construction, we first determine the number of totally split rational places in each case. We then compare these values with the polynomial benchmark $M_q(r+1)$ defined above.

\begin{itemize}
\item $r=2$. Let $w\in\mathbb{F}_q^*$ and consider the rational function
\[
h(x)=\frac{x^3-3w^2x+w^3}{x(x-w)}.
\]
A direct computation shows the transformation $\phi(x)=w^2/(w-x)$ has order three and satisfies $h(\phi(x))=h(x)$. Thus $\mathbb{F}_q(x)/\mathbb{F}_q(h(x))$ is cyclic of degree three. Its short orbits are the fixed points of $\phi$, which satisfy $x^2-wx+w^2=0$. Corollary~\ref{good 3} then gives
\[
\#T^1_{\mathrm{split}}(h)=
\begin{cases}
\dfrac{q}{3}, & q=3^m,\\[2ex]
\dfrac{q-1}{3}, & q\equiv1\pmod3,\\[2ex]
\dfrac{q+1}{3}, & q\equiv2\pmod3.
\end{cases}
\]

When $q\equiv2\pmod3$, the results of
\cite{Chen Mesnager zhao2021,Duke F M2022} show that
the extension associated with any cubic polynomial
is not Galois, and its Galois closure has Galois group
of order six. More precisely, the maximum number of
totally split rational places among all cubic good polynomials is
\[
M_q(3)=
\begin{cases}
\dfrac{q-2}{6}, & q\equiv2\pmod6,\\[2ex]
\dfrac{q+1}{6}, & q\equiv5\pmod6.
\end{cases}
\]
Thus, the maximum length of LRCs obtained from cubic
good polynomials through the Tamo--Barg construction is
$3M_q(3)\le(q+1)/2$. In contrast, the rational function above has $(q+1)/3$ totally split rational places,
and hence yields optimal LRCs of length $q+1$ by
Theorem~\ref{main construction}.

\item $r=3$. Let $q$ be odd, $d\in\mathbb{F}_q^*$, and consider the rational function
\[
h(x)=\frac{4x^4-12d^2x^2+8d^3x-d^4}
{2x(x-d)(2x-d)}.
\]
One readily verifies that the transformation
$\phi(x)=d^2/(2(d-x))$ has order four and satisfies
$h(\phi(x))=h(x)$. Then
$\mathbb{F}_q(x)/\mathbb{F}_q(h(x))$ is cyclic of degree four.
Both equations $\phi(x)=x$ and $\phi^2(x)=x$ reduce to
$2x^2-2dx+d^2=0$, whose discriminant is $-4d^2$.
Hence there are no orbits of size two, and the number of
fixed points is two if $q\equiv1\pmod4$ and zero if
$q\equiv3\pmod4$. By Corollary~\ref{good 3}, we obtain
\[
\#T^1_{\mathrm{split}}(h)=
\begin{cases}
\dfrac{q-1}{4}, & q\equiv1\pmod4,\\[2ex]
\dfrac{q+1}{4}, & q\equiv3\pmod4.
\end{cases}
\]

When $q\equiv3\pmod4$ and $q>3$, the results of
\cite{Chen Mesnager zhao2021} give the maximum number
of totally split rational places among all quartic
good polynomials as
\[
M_q(4)=
\begin{cases}
\dfrac{q-3}{8}, & q\equiv3\pmod8,\\[2ex]
\dfrac{q+1}{8}, & q\equiv7\pmod8.
\end{cases}
\]
Consequently, the maximum length of LRCs obtained
from quartic good polynomials through the Tamo--Barg
construction is $4M_q(4)\le(q+1)/2$.
On the other hand, the rational function above has
$(q+1)/4$ totally split rational places, yielding
optimal LRCs of length $q+1$ by
Theorem~\ref{main construction}.

\end{itemize}

The preceding two cases are based on Galois extensions. For
locality $r=4$, however, such extensions do not always exist.
In particular, when $q\equiv2,3\pmod5$, we have
$5\nmid q(q^2-1)$, thus $\mathrm{PGL}(2,q)$ contains no subgroup
of order $5$. Consequently, no rational function of degree $5$ induces a Galois extension of $\mathbb{F}_q(x)$.
Nevertheless, we show that a non-Galois rational function
can still have more totally split rational places than
every polynomial of the same degree.

\begin{itemize}
\item $r=4$. Suppose that $\operatorname{char}(\mathbb{F}_q)>5$
and $q\equiv2,3\pmod5$. Consider the rational function
\[
h(x)=\frac{x^5-2x^4+3x^3-2x+1}{x^2(x-1)^2}
=x+\frac{1}{x^2}+\frac{2}{x-1}+\frac{1}{(x-1)^2}.
\]
For this function, the extension
$\mathbb{F}_q(x)/\mathbb{F}_q(h(x))$ is separable but not
Galois, and we have
\[
\#T^1_{\mathrm{split}}(h)\ge
\frac{q+1-2\sqrt q-20}{10}.
\]
Moreover,
\[
M_q(5)\le\frac{q+1+98\sqrt q}{120}.
\]
In particular, $\#T^1_{\mathrm{split}}(h)>M_q(5)$ whenever
$q\ge163$.

\begin{proof}
The poles of $h$ at $0,1,\infty$ have orders $2,2,1$,
respectively.
Since the ramification indices above $\infty$ are not
equal, this extension is not Galois.

To estimate $\#T^1_{\mathrm{split}}(h)$, consider the
elliptic curve
\[
E:\quad y^2-y=x^3-x^2.
\]
Its discriminant is $-11\ne0$ under our assumptions (indeed, characteristic $11$ is excluded by $q\equiv2,3\pmod5$),
and $P=(0,0)$ is a point of order five. For a general point $R=(x,y)$ on $E$, the addition law gives
\[
h(x(R))=\sum_{j=0}^{4}x(R+jP)-2,
\qquad x(R+P)=\frac{y}{x^2}.
\]
Thus $h$ is invariant under translation by $P$ and
the involution $R\mapsto-R$. These automorphisms
generate a dihedral group $D_{10}$ of order ten.
Since
\[
[\mathbb{F}_q(E):\mathbb{F}_q(h(x))]
=2\deg(h)=10,
\]
the extension $\mathbb{F}_q(E)/\mathbb{F}_q(h(x))$
is Galois with group $D_{10}$. Moreover, the
conjugates $x$ and $x(R+P)=y/x^2$ generate
$\mathbb{F}_q(E)$. Hence this field is precisely
the Galois closure of
$\mathbb{F}_q(x)/\mathbb{F}_q(h(x))$.

We now consider the action of $D_{10}$ on
$E(\mathbb{F}_q)$. Nonidentity translations have
no fixed points, while each of the five involutions
$R\mapsto jP-R$ fixes at most four rational points,
since its fixed points satisfy $2R=jP$.
Consequently, at least $\#E(\mathbb{F}_q)-20$
rational points lie in orbits of size ten, each
corresponding to a totally split rational place
of $h$. By the Hasse--Weil bound,
\[
\#T^1_{\mathrm{split}}(h)
\ge\frac{\#E(\mathbb{F}_q)-20}{10}
\ge\frac{q+1-2\sqrt q-20}{10}.
\]

We next consider quintic good polynomials.
By the results of \cite{Duke F M2022}, the Galois group of
the Galois closure associated with any such
polynomial $f$ is $S_5$ under our assumptions.
Since the characteristic is greater than five,
all ramification is tame. The ramification
contribution at infinity is four, and the total
contribution at finite points is also four.
For an inertia permutation of order $e$ with
$c$ cycles, one has
\[
1-\frac1e\le\frac{5-c}{2}.
\]
Thus the total finite contribution to the
Riemann--Hurwitz formula for the Galois closure
is at most two. As the contribution at infinity
is $4/5$, we obtain
\[
2g_f-2\le120\left(-2+\frac45+2\right)=96,
\]
and hence $g_f\le49$. Since $f$ is good, its
Galois closure has constant field $\mathbb{F}_q$.
Applying Corollary~\ref{estimate}, we obtain
\[
M_q(5)\le\frac{q+1+98\sqrt q}{120}.
\]

Finally,
\[
\begin{aligned}
\#T^1_{\mathrm{split}}(h)-M_q(5)
&\ge\frac{q+1-2\sqrt q-20}{10}
-\frac{q+1+98\sqrt q}{120}\\
&=\frac{11q-122\sqrt q-229}{120}>0
\end{aligned}
\]
for $q\ge163$. This completes the proof.
\end{proof}

The advantage also occurs over smaller fields. For example,
over $\mathbb{F}_{67}$, the above function has $7$ totally
split rational places, while $M_{67}(5)=4$. Taking $k=20$
in Theorem~\ref{main construction}, we obtain an optimal
$[35,20,12]_{67}$-LRC of locality $4$, whereas quintic
good polynomials yield codes of length at most $20$
through the Tamo--Barg construction. Using Magma, we
found further good rational functions of degree five
with more totally split rational places than any quintic
good polynomial over the same field. These examples
are listed in Table~\ref{tabxx}, together with the
exact values of $M_q(5)$ for comparison. The fiber counts include the point at infinity, and the values of $M_q(5)$ were checked by exhaustive enumeration of quintic polynomials up to affine changes of variables.

\end{itemize}

\begin{table}[H]
\scriptsize
\centering
\setlength{\tabcolsep}{3mm}
\caption{Rational functions of degree five with more totally split
rational places than all quintic polynomials over the same field.}
\label{tabxx}
\begin{tabular}{cccc}
\toprule
$q$ & $M_q(5)$ & Rational function $h$
& $\#T^1_{\mathrm{split}}(h)$\\
\midrule
$37$ & $2$ &
$\displaystyle\frac{x^5+x^3+18x^2}
{(x-1)(x-2)(x-3)(x-11)}$ & $4$\\
\midrule
$37$ & $2$ &
$\displaystyle\frac{x^5+36x^4+15x^3+16x^2+31x+33}
{(x-3)(x-4)(x-5)(x-30)}$ & $4$\\
\midrule
$37$ & $2$ &
$\displaystyle\frac{x^5+15x^4+36x^3+13x^2+13x+28}
{(x-11)(x-23)(x-27)(x-31)}$ & $4$\\
\midrule
$43$ & $3$ &
$\displaystyle\frac{x^5+10x^4+36x^3+7x^2+12x+23}
{(x-2)(x-19)(x-24)(x-36)}$ & $5$\\
\midrule
$43$ & $3$ &
$\displaystyle\frac{x^5+6x^4+35x^3+25x^2+7x+27}
{(x-5)(x-18)(x-33)(x-42)}$ & $5$\\
\midrule
$47$ & $3$ &
$\displaystyle\frac{x^5+x^4+22x^3+37x^2+28x+19}
{(x-1)(x-10)(x-20)(x-21)}$ & $6$\\
\midrule
$47$ & $3$ &
$\displaystyle\frac{x^5+4x^4+32x^3+25x^2+15x+34}
{(x-18)(x-23)(x-42)(x-44)}$ & $5$\\
\midrule
$53$ & $3$ &
$\displaystyle\frac{x^5+8x^2}
{(x-1)(x-2)(x-3)(x-49)}$ & $5$\\
\midrule
$53$ & $3$ &
$\displaystyle\frac{x^5+15x^4+40x^3+32x^2+13x+18}
{x(x-2)(x+1)(x+2)}$ & $5$\\
\midrule
$53$ & $3$ &
$\displaystyle\frac{x^5+22x^4+16x^3+42x^2+15x+20}
{(x-22)(x-35)(x-47)(x-51)}$ & $5$\\
\midrule
$53$ & $3$ &
$\displaystyle\frac{x^5+42x^4+3x^3+14x^2+9x+47}
{(x-16)(x-18)(x-24)(x-35)}$ & $5$\\
\midrule
$67$ & $4$ &
$\displaystyle\frac{x^5+3x^4+53x^3+24x^2+2x+27}
{(x-16)(x-17)(x-50)(x-58)}$ & $6$\\
\midrule
$73$ & $4$ &
$\displaystyle\frac{x^5+61x^4+3x^3+60x^2+2x+24}
{(x-15)(x-22)(x-40)(x-65)}$ & $6$\\
\midrule
$73$ & $4$ &
$\displaystyle\frac{x^5+18x^4+52x^3+4x^2+17x+42}
{(x-5)(x-42)(x-47)(x-49)}$ & $6$\\
\midrule
$97$ & $5$ &
$\displaystyle\frac{x^5+46x^3+76x^2}
{(x-1)(x-2)(x-3)(x-67)}$ & $7$\\
\midrule
$97$ & $5$ &
$\displaystyle\frac{x^5+39x^4+80x^3+71x^2+44x+91}
{(x-18)(x-32)(x-49)(x-58)}$ & $7$\\
\midrule
$103$ & $5$ &
$\displaystyle\frac{x^5+23x^3+97x}
{(x-1)(x-2)(x-3)(x-81)}$ & $8$\\
\bottomrule
\end{tabular}
\end{table}

\section{Conclusion}\label{summary}
In this paper, we generalized the notion of good polynomials to \emph{good rational functions} and developed a construction of optimal LRCs based on rational functions over the projective line. By studying the associated function field extensions and their Galois closures, we established an exact counting formula and explicit bounds for the number of totally split rational places, together with sharp bounds for ramification and, under the stated tameness hypothesis, for the genus of the Galois closure. We also provided explicit constructions arising from both Galois and non-Galois extensions. These examples showed that good rational functions can yield optimal LRCs with lengths exceeding those attainable from good polynomials of the same degree.

This work opens several promising research directions. A natural continuation is the systematic classification of good rational functions of low degree, together with a precise determination of their associated Galois groups. This direction parallels recent advances for polynomials in \cite{Chen Mesnager zhao2021} and \cite{Duke F M2022}, and may lead to a deeper understanding of the interplay between algebraic structures and coding parameters. Another direction is to investigate systematic constructions from further subgroups of $\mathrm{PGL}(2,q)$ and to compare their full rational fibers with those of existing automorphism-group constructions. Finally, it would be of interest to investigate potential applications of good rational functions in related areas, such as algebraic geometry codes, cryptographic constructions, or pseudorandom structures, where controlled splitting behavior plays a crucial role.

\end{document}